\documentclass[letterpaper, 10 pt, conference]{ieeeconf}  % Comment this line out if you need a4paper

\UseRawInputEncoding

\IEEEoverridecommandlockouts                              % This command is only needed if 
\usepackage{color}

\usepackage{amsmath} % assumes amsmath package installed
\usepackage{amssymb}  % assumes amsmath package installed
\usepackage{cite}
\usepackage{footmisc}

\usepackage{algorithm}  
\usepackage{algpseudocode}  
\usepackage{booktabs}
\usepackage{makecell}
\usepackage{graphicx}
\usepackage{threeparttable}
\usepackage{cuted}

\newcommand{\bs}[1]{\boldsymbol{#1}}

\newcommand{\Aut}{\text{Aut}}
\newcommand{\Inn}{\text{Inn}}
\newcommand{\Out}{\text{Out}}
\newcommand{\Ad}{\text{Ad}}
\newcommand{\ad}{\text{ad}}
\newcommand{\TFG}{\text{TFG}}
\newcommand{\MFG}{\text{MFG}}
\newcommand{\SO}[1]{\text{SO}(#1)}
\newcommand{\diag}[1]{\text{diag}(#1)}

\newcommand{\SSS}{\boldsymbol{\mathfrak{S}}}
\newcommand{\WWW}{\boldsymbol{\mathfrak{W}}}

\newtheorem{definition}{Definition}

\newtheorem{remark}{Remark}

\newtheorem{lemma}{Lemma}
\newtheorem{theorem}{Theorem}

\usepackage{xcolor}
\usepackage{xpatch}

\makeatletter
\ExplSyntaxOn
\cs_new:Npn \bibColoredItems #1#2
  {
    \clist_map_inline:nn {#2} { \cs_new:cpn {bib@colored@##1} {#1} } 
  }
\ExplSyntaxOff

\newcommand\bib@setcolor[1]{%
  \ifcsname bib@colored@#1\endcsname
    \expandafter\color\expandafter{\csname bib@colored@#1\endcsname}
  \else
    \normalcolor
  \fi
}

\xpatchcmd\@bibitem
  {\item}
  {\bib@setcolor{#1}\item}
  {}{\fail}

\xpatchcmd\@lbibitem
  {\item}
  {\bib@setcolor{#2}\item}
  {}{\fail}
\makeatother
\title{\LARGE \bf
Classification of Linear Observed Systems on Multi-Frame Groups \\via Automorphisms
}

\author{Changwu Liu$^{1}$ and Yuan Shen$^{2}$% <-this % stops a space
\thanks{$^{1}$Changwu Liu is with Department of Electronic Engineering, Tsinghua University, Beijing, China
        {\tt\small liucw\_ee@tsinghua.edu.cn}}%
\thanks{$^{2}$Yuan Shen is with Department of Electronic Engineering, Tsinghua University, Beijing, China
        {\tt\small shenyuan\_ee@tsinghua.edu.cn}}%      
}

\begin{document}

\maketitle
\thispagestyle{empty}
\pagestyle{empty}

%%%%%%%%%%%%%%%%%%%%%%%%%%%%%%%%%%%%%%%%%%%%%%%%%%%%%%%%%%%%%%%%%%%%%%%%%%%%%%%%
\begin{abstract}

Many navigation problems can be formulated as observer design on linear observed systems with a two-frame group structure, on which an invariant filter can be implemented with guaranteed consistency and stability. It's still unclear how this could be generalized to simultaneous estimation of the poses of multiple frames and the general forms of the linear observed systems involving multiple frames remain unknown. In this letter, we propose a multi-frame group structure by semi-direct product using the two-frame group as building blocks, covering all natural extensions. More importantly, we give a systematic direct calculation to classify all possible forms of linear observed systems including process ODEs and algebraic observations on such multi-frame group through its automorphism structure, in comparison to the existing classification on two-frame groups relying on ingenious construction. Depth-camera inertial odometry with online extrinsics calibration is provided as an application.    

\end{abstract}


%%%%%%%%%%%%%%%%%%%%%%%%%%%%%%%%%%%%%%%%%%%%%%%%%%%%%%%%%%%%%%%%%%%%%%%%%%%%%%%%
\section{Introduction}

\PARstart{L}{inear} observed systems, first formulated on Lie groups\cite{LoS}, then generalized to manifolds with connection\cite{LoS_Mfd}, comprise a family of nonlinear systems on non-Euclidean space compatible with some geometric-structure preserving transformations of the state space\cite{LoS_Mfd}. Careful observer design on a linear observed system has promising properties provided the geometry of the problem is respected, and thus those systems are of particular interest to the nonlinear system and observer community. Given a linear observed system on Lie group, an invariant extended Kalman-filter (IEKF) can be implemented with exact error propagation in the process ODE with state-independent Jacobians\cite{InEKF, AnnuRevInEKF, LoS}. Such properties guarantee the uniform boundedness of the error evolution of IEKF, leading to non-trivial local stability results using mere assumptions on the true trajectory compared to traditional EKFs\cite{InEKF}. Equivariant filters (EqF)\cite{EqF,EqFCDC}, designed for systems on homogeneous spaces through an equivariant lift of the system to its Lie symmetry group via input extensions\cite{EqSys}, are also applicable to linear observed systems with guaranteed consistency and local convergence\cite{EqF_SysLieGroup}. Noise models reflecting the stochastic characteristics of sensors can be easily associated to either IEKF or EqF, making it popular for practitioners in robotics. There are a wide range of successful implementations of IEKF or EqF in visual-inertial based multi-sensor navigation, like \cite{EqVIO, MSCEqF, Decoupled-RIEKF, InGVIO,ExploitSymm}, to name a few. Besides observer design based on linearization (e.g. IEKF and EqF), constructive observer design techniques on linear observed systems are eligible for almost-global convergence beyond local results\cite{AutoErrorGroupAffine, ConstructiveInertialNav, SynchronousObsInertial}. The error dynamics of a general linear observed system is state-independent but may depend on control inputs. The key-observation in \cite{AutoErrorGroupAffine, ConstructiveInertialNav} is that the error evolution could be cast into a fully autonomous form through augmentation of the original Lie group state with the automorphism group of its Lie algebra, leading to convergence results in the large. Observer design methods on linear observed systems by linearization or construction are of equal importance, since the former has accuracy advantage in long-term operation and the latter may be used as an initializer for the former.          

Prior to adopting one of the above observer frameworks for linear observed systems, one has to first model the equations obtained from physical laws as linear observed systems. It turns out that this is highly nontrivial in practice. For instance, the navigation problem of a single vehicle, roughly speaking, is the state estimation of attitude, position and velocity between frames. The non-biased IMU kinematics is canonically formulated on a smooth state manifold diffeomorphic to $\text{SO}(3)\times\mathbb{R}^6$. Note that various group structures can be endowed on $\text{SO}(3)\times\mathbb{R}^6$. One has to ingeniously choose the $\text{SE}_2(3)$\cite{InEKF} and later verify the establishment of group affine property to confirm the feasibility of modelling by linear observed systems. This becomes intractable when either the state structure or the system equation is more complicated than the case above, even the existence of such structure making the system linear observed remains unknown. The invention of the two-frame group (TFG) is a pioneer work tackling this issue by A. Barrau and S. Bonnabel\cite{TFG}. In contrast to the above procedure, they attack the modelling problem by first considering a generic group structure covering most navigation applications and construct the general forms of systems on this group with geometric interpretations, i.e. the natural two-frame systems. Though there are emerging applications of employing the TFG art\cite{ActiveSensingTFG, TFGwheel,TFGPseudoMeas}, it's still unclear about the generic group structure of simultaneous estimation of states involving multiple coupled frames. Inspired by \cite{TFG}, we develop a multi-frame group (MFG) structure by semi-direct product using the two-frame group as building blocks, covering all natural extensions. It's of theoretical and practical interest to classify all linear observed systems on MFG, i.e. explicitly showing all possible forms of process ODE and algebraic observation equations. We emphasize on the term `classification': it rules out the possibility of a particular system to fall into the scope of linear observed system under some state group construction as long as the system does not match the form in the classification theorems. Note this is not explicitly proved in \cite{TFG}. Unlike delicate geometric construction\cite{TFG}, a systematic approach is proposed to directly calculate all possible forms of the linear observed systems on MFG including process ODEs and algebraic observations via the automorphism group of MFG. Based on our classification, practitioner is only required to check whether the sensor system equations involving motions of multiple rigid bodies are in a form covered by our theorems. If yes, one might naturally inherit a group structure from MFG. If no, observer techniques on linear observed systems may not be applicable. 

Our major contributions are summarized as:
\begin{itemize}
\item a multi-frame group (MFG) structure is constructed by semi-direct product using the two-frame group as building blocks, covering natural extensions;
\item a systematic approach utilizing the automorphism group is proposed to classify linear observed systems on MFG including process ODEs and algebraic observations;
\item depth-camera inertial odometry with online extrinsics calibration serves as an application of the theory. 
\end{itemize}

\section{Preliminaries}\label{sec::pre}

Mathematical preliminaries regarding Lie groups and algebras used in this letter are briefly introduced. Readers may refer to \cite{IntroSmoothMfd,LieGroupAlgRep} for a comprehensive exposition.

A Lie group $G$ is at the same time a smooth manifold and a group\cite{IntroSmoothMfd}. The push-forward of the left multiplication induces a left translation on the tangent bundle $TG$ of $G$. The left-invariant vector fields $L(G)$ are vector fields on $G$ invariant under the action of left translations. $L(G)$ happens to be a finite-dimensional $\mathbb{R}$-vector space due to the one-to-one correspondence between $L(G)$ and the tangent space $T_{id}G$ at the identity. This identification gives $T_{id}G$ a Lie algebra $\mathfrak{g}$ structure from the differential geometric bracket of $L(G)$. The exponential map $\exp:\mathfrak{g}\rightarrow G$ establishes a diffeomorphism between $\mathfrak{g}$ and an open subset of $G$ containing the identity. Matrix Lie groups form a subset of abstract Lie groups through embedding into some matrix space $\mathbb{C}^{n\times n}$, and their multiplications and left-translations are matrix products. Meanwhile, the Lie algebra are matrices of the same size and the Lie bracket is the matrix commutator. $\exp$ is realized by matrix exponential.

The structure preserving map of a Lie group $G$ is the automorphism $\Psi:G\rightarrow G$, where $\Psi$ is smooth, injective and surjective. Moreover, it preserves group multiplication, i.e. $\Psi(g_1g_2)=\Psi(g_1)\Psi(g_2),\forall g_1,g_2\in G$. All $\Psi$s compose an abstract group under composition, denoted the automorphism group $\Aut(G)$. In most cases, $\Aut(G)$ itself can be given a Lie structure. Every $\Psi$ induces a Lie algebra automorphism, denoted $\psi:\mathfrak{g}\rightarrow\mathfrak{g}$, which is a linear isomorphism preserving the Lie bracket. All $\psi$s again form a Lie group, easily seen to be a subgroup of $\text{GL}(\mathfrak{g})$. $\psi$ corresponds to a unique $\Psi$ as long as $G$ is simply-connected, meaning $\Aut(G)\stackrel{\sim}{=}\Aut(\mathfrak{g})$. Transformations on $G$ in the form of conjugation $\mathcal{C}_{g}:G\rightarrow G,\ h\mapsto ghg^{-1},\ \forall g,h\in G,$ form a normal subgroup of $\Aut(G)$, termed the inner-automorphisms denoted $\Inn(G)$. The elements in the quotient group from $\Out(G)\stackrel{\sim}{=}\Aut(G)/\Inn(G)$ are outer-automorphisms. Note $\mathcal{C}:G\rightarrow\Inn(G),\ g\mapsto\mathcal{C}_g$ can be interpretated as a homomorphism from $G$ to $\Aut(G)$. The derivative of the conjugation at the identity defines the adjoint map $\Ad_g:\mathfrak{g}\rightarrow\mathfrak{g}$ by $\Ad_g:=(d\mathcal{C}_g)_{id}$. $\Ad$s form a Lie subgroup of $\Aut(\mathfrak{g})$, whose Lie algebra $\ad\in\text{End}(\mathfrak{g})$ is the bracket on $\mathfrak{g}$.

Given group $N$ and $Q$, group $G$ is said to be the extension of $N$ by $Q$ if $N$ is a normal subgroup of $G$ and $Q$ is isomorphic to the quotient group $G/N$, meaning $G$ takes value in the product space $Q\times N$. In practice, this could be realized by a semi-direct product, namely by a homomorphism $\varphi:Q\rightarrow\Aut(N)$, and one can construct the group law on $G$. Semi-direct product via $\varphi$ is denoted by $G=Q\ltimes_\varphi N$. Note $\varphi=id_N$ reproduces direct-product.

Bold letters are for vectors and matrices. $\bs{0}$ and $\bs{I}$ denote the zero and identity matrix respectively. 

\section{Construction of Multi-Frame Groups}\label{sec::mfg}

Rotation, as a length and orientation preserving linear transformation between the body frame and the world frame, takes value in $\text{SO}(d)$. In navigation problems on a two-frame setting, one usually estimates the rotation $\text{SO}(d)$ along with some $\mathbb{R}^d$-valued vectors associated with either the body or the world frame\cite{TFG}, rigorously defined by group extension.

\begin{definition}\label{def::type_I_extension}
  Let $G,H$ be two groups. Denote $\phi:G\rightarrow\Aut(H)$ to be a homomorphism from $G$ to the automorphism group of $H$. The set $G\times H^n\times H^m$ ($n,m\in\mathbb{N}$) can be given a group structure by semi-direct product, and the corresponding group multiplication is defined by
  \begin{align*}
    &(g^{(1)},h_1^{(1)},\cdots, h_n^{(1)},h_{n+1}^{(1)},\cdots, h_{n+m}^{(1)})\\
    &\cdot(g^{(2)},h_1^{(2)},\cdots, h_n^{(2)},h_{n+1}^{(2)},\cdots, h_{n+m}^{(2)})\\
    =\ &(g^{(1)}g^{(2)},h_1^{(1)}\phi_{g^{(1)}}(h_1^{(2)}),\cdots, h_n^{(1)}\phi_{g^{(1)}}(h_n^{(2)}),\\
    &\phi^{-1}_{g^{(2)}}(h_{n+1}^{(1)})h_{n+1}^{(2)},\cdots, \phi^{-1}_{g^{(2)}}(h_{n+m}^{(1)})h_{n+m}^{(2)})
  \end{align*}
  where $g^{(\cdot)}\in G$ and $h_{\cdot}^{(\cdot)}\in H$. The new group is termed \textbf{the type-I semi-direct product extended group} of $G,H$ by $\phi$ with parameter $n,m$.
\end{definition}

$\TFG(d,n,m)$ is isomorphic to the type-I semi-direct product extended group of $G=\SO{d}$ and $H=\mathbb{R}^d$ by $\phi$ with parameter $n,m$, where $\phi$ is the canonical action of $\SO{d}$ on $\mathbb{R}^d$. It's straightforward to verify that $\TFG(d,n,m)$ has a matrix embedding in the form of \eqref{eq::tfg_embedding}
\begin{equation}
  \label{eq::tfg_embedding}
  \bs{T}=\begin{bmatrix}
    \bs{R} & \bs{r}\\
    \bs{0} & \bs{I}
    \end{bmatrix}\in\text{GL}(d+n+m,\mathbb{R}),\ 
    \bs{r}=\begin{bmatrix}
      \bs{x} & \bs{Ry}
    \end{bmatrix}
\end{equation}
where the matrix product coincides with the multiplication defined above. Note that $\bs{R}\in\SO{d},\bs{r}\in\mathbb{R}^{d\times(n+m)},\bs{x}\in\mathbb{R}^{d\times n},\bs{y}\in\mathbb{R}^{d\times m}$. Each column of $\bs{x}$ or $\bs{y}$ is the world-frame or body-frame related vector state respectively. To describe non-trivial coupling of multiple frames, we need another type of group extension.

\begin{definition}\label{def::type_II_extension}
  Let $G$ be an abstract group and $\phi$ be a homomorphism from $G$ to the automorphism group of $G$, i.e. $\phi:G\rightarrow\Aut(G)$. The set $G\times(G)^s\times(G)^t$ ($s,t\in\mathbb{N}$) can be given a group structure by semi-direct product and the corresponding group multiplication is defined by
  \begin{align*}
    &\ (g_1^{(1)},g_2^{(1)},\cdots, g_{s+1}^{(1)},g_{s+2}^{(1)},\cdots, g_{s+t+1}^{(1)})\\
    &\ \cdot(g_1^{(2)},g_2^{(2)},\cdots, g_{s+1}^{(2)},g_{s+2}^{(2)},\cdots, g_{s+t+1}^{(2)})\\
    =&\ (g_1^{(1)}g_1^{(2)},
    g_2^{(1)}\phi_{g_1^{(1)}}(g_2^{(2)}),g_3^{(1)}\phi_{g_2^{(1)}}\circ\phi_{g_1^{(1)}}(g_3^{(2)}),\\
    &\ \cdots, g_{s+1}^{(1)}\phi_{g_{s}^{(1)}}\circ\phi_{g_{s-1}^{(1)}}\circ\cdots\circ\phi_{g_1^{(1)}}(g_{s+1}^{(2)}),\\
    &\ \phi^{-1}_{g^{(2)}_1}(g_{s+2}^{(1)})g_{s+2}^{(2)},\phi^{-1}_{g^{(2)}_{s+2}}\circ\phi^{-1}_{g^{(2)}_1}(g_{s+3}^{(1)})g_{s+3}^{(2)},\\
    &\ \cdots, \phi^{-1}_{g^{(2)}_{s+t}}\circ\cdots\circ\phi^{-1}_{g^{(2)}_{s+2}}\circ\phi^{-1}_{g_1^{(2)}}(g_{s+t+1}^{(1)})g_{s+t+1}^{(2)})
  \end{align*}
  where $g_{\cdot}^{(\cdot)}\in G$. This new group is termed \textbf{the type-II semi-direct product extended group} of $G$ by $\phi$ with $s,t$.
\end{definition}

The automorphism group of TFG is in the following form,
\begin{equation*}
  \Aut(\TFG)=\left\{\psi_{\bs{S}}\bigg|\psi_{\bs{S}}(\bs{T})=\bs{S}\bs{T}\bs{S}^{-1},\bs{S}=\begin{bmatrix}
    \bs{\Omega} & \bs{\varrho}\\
    \bs{0} & \bs{A}
  \end{bmatrix}\right\}
\end{equation*}
where $\bs{T}\in\TFG(d,n,m)$, $\bs{\Omega}\in\SO{d}$, $\bs{\varrho}\in\mathbb{R}^{d\times(n+m)}$ and $\bs{A}\in\text{GL}(n+m,\mathbb{R})$. $\bs{S}$ itself composes a matrix Lie group termed $\text{SIM}_{n+m}(d)$\cite{ConstructiveInertialNav}. A canonical choice of the homomorphism $\phi:\TFG\rightarrow\Aut(\TFG)$ is $\phi(\bs{T})=\psi_{\bs{T}}$, i.e. the mapping into $\Inn(\TFG)$. The multi-frame group is thereafter defined.

\begin{definition}
  The \textbf{multi-frame group} (MFG), denoted $\MFG(d,n,m,s,t)$, is the type-II semi-direct product extended group of $\TFG(d,n,m)$ by the above $\phi$ with $s,t\in\mathbb{N}$.
\end{definition}

\begin{theorem}\label{th::mfg_embedding}
  Every element $\bs{\chi}$ in a multi-frame group $\MFG(d,n,m,s,t)$ has a matrix group embedding of block-diagonal matrices with $\TFG(d,n,m)$-valued blocks, as: 
  \begin{gather*}
    \begin{split}
      \bs{\chi}=\text{diag}\bigl(&\bs{T}_0, {^l}\bs{T}_1\bs{T}_0,{^l}\bs{T}_2{^l}\bs{T}_1\bs{T}_0,...,{^l}\bs{T}_s{^l}\bs{T}_{s-1}\cdots\bs{T}_0,\\
      &\bs{T}_0{^r}\bs{T}_1,\bs{T}_0{^r}\bs{T}_1{^r}\bs{T}_2,...,\bs{T}_0\cdots{^r}\bs{T}_{t-1}{^r}\bs{T}_t\bigl)
    \end{split}
  \end{gather*}
  where each ${^{(\cdot)}}\bs{T}_{(\cdot)}$ takes value in TFG.
\end{theorem}
\begin{proof}
  Matrix product coincides with Definition~\ref{def::type_II_extension}. 
\end{proof}

Note the upper-left superscript $l$ or $r$ is short for multiplication on the `left' or `right' of the core state $\bs{T}_0$. A simplified notation for a MFG member is \eqref{eq::mfg_simple_notation}.
\begin{equation}
  \label{eq::mfg_simple_notation}
  \bs{\chi}=\left(\bs{T}_0\bigl\vert{^l}\bs{T}_j\cdots\bs{T}_0\bigl\vert\bs{T}_0\cdots{^r}\bs{T}_j \right)
\end{equation}
An intuitive geometric picture of the transformation chain is illustrated in Fig.~\ref{fig::mfg_chain_structure}. It supports the coupling of multiple frames by the inner-automorphism $\phi:\TFG\rightarrow\Aut(\TFG)$ rather than the trivial $\phi$ being identity, because the latter breaks the coupling in the chain and decomposes the estimation problem into several independent systems.  

\begin{figure}[t]%\color{blue}
  \centering
  \parbox{3in}{
    \centering
    \includegraphics[scale=0.60]{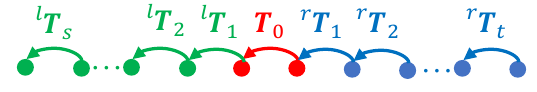}
  }
  \caption{Transformation chain structure of MFG. Each node represents a frame. Each arrow is a TFG-valued transformation between two nodes. All transformations are to be estimated simultaneously.}
  \label{fig::mfg_chain_structure}
\end{figure} 

Before delving into the classification of linear observed systems on MFG, we derive the nonlinear update formula for filters on MFG. The vector-valued error state of MFG is defined in its Lie algebra. Let the exponential of the error state be $\delta\bs{\chi}=\text{diag}\left(\delta\bs{T}_0\vert\delta{^l}T_j\vert{\delta}{^r}\bs{T}_j\right)\in\MFG$. The explicit formula for the MFG exponential is omitted since each of its diagonal blocks is the exponential of TFG\cite{TFG}. Let $\hat{\bs{\chi}}\in\MFG$ be the estimated state and $\bs{\chi}\in\MFG$ be the true state. The update equations for the left-invariant error $\delta\bs{\chi}=\bs{\chi}^{-1}\hat{\bs{\chi}}$ are \eqref{eq::mfg_li_error1}--\eqref{eq::mfg_li_error2} in terms of ${^{(\cdot)}}\bs{T}_{(\cdot)}$. The right-invariant case $\delta\bs{\chi}=\hat{\bs{\chi}}\bs{\chi}^{-1}$ is similar and thus omitted to save space.
\begin{align}
  \label{eq::mfg_li_error1}
  \bs{T}_0&=\hat{\bs{T}}_0\delta\bs{T}_0,\quad {^r}\bs{T}_j=(\delta {^r}\bs{T}_{j-1})^{-1}{^r}\hat{\bs{T}}_j\delta {^r}\bs{T}_j \\
  \label{eq::mfg_li_error2}
  {^l}\bs{T}_j&={^l}\hat{\bs{T}}_j\mathcal{C}_{{^l}\hat{\bs{T}}_{j-1}}\circ\mathcal{C}_{{^l}\hat{\bs{T}}_{j-2}}\cdots\mathcal{C}_{\hat{\bs{T}}_{0}}\left[\delta {^l}\bs{T}_j({^l}\delta\bs{T}_{j-1})^{-1}\right]
\end{align}

The final and most important preparation for the classification is the structure of the automorphism group of MFG.
\begin{theorem}\label{th::mfg_automorphism}
  The automorphism group of MFG is:
  \begin{equation*}
    \Aut(\MFG)=\left\{\psi_{\bs{S}}\bigl|\psi_{\bs{S}}(\bs{\chi})=\bs{S}\bs{\chi}\bs{S}^{-1}\right\}
  \end{equation*}
  where $\bs{S}=\text{diag}(\bs{S}_1,...,\bs{S}_{s+t+1})$ and each diagonal block $\bs{S}_i$ has the structure $\bs{S}_i=\begin{bmatrix}
    \bs{\Omega}_i & \bs{\varrho}_i\\
    \bs{0} & \bs{A}_i
  \end{bmatrix}\in\text{SIM}_{n+m}(d)$ with $\bs{\Omega}_i\in\SO{d}$, $\bs{\varrho}_i\in\mathbb{R}^{d\times(n+m)}$ and $\bs{A}_i\in\text{GL}(n+m,\mathbb{R})$.
\end{theorem}
\begin{proof}
  Direct calculation verifies the theorem.
\end{proof}

\section{Linear Observed Systems on MFGs}\label{sec::los_mfg}

A linear observed system on a Lie group consists of a continuous-time dynamics whose flow is compatible with the automorphism of the state group and an algebraic observation equation in the form of some group action of the state on some constant vector\cite{LoS,AnnuRevInEKF}.

\subsection{Classification of Group Affine Dynamics on MFGs}

The final goal of the classification is to give explicit formulas of all possible forms of group affine dynamics in terms of $\SO{d}$ and vector-valued states associated to each frame. Let $\Phi^t(\bs{\chi}_0):\mathbb{R}\times\MFG\rightarrow\MFG$ be the flow of a group affine dynamics, i.e. a global one-parameter group of local diffeomorphisms\cite{LoS_Mfd}, then $\Phi^t(\bs{\chi}_0)=\psi^{u_t}(\bs{\chi}_0)\Phi^t(id)$, where $\psi^{u_t}$ is an arbitrary trajectory depending on control inputs $u_t$ in $\Aut(\MFG)$ and $\Phi^t(id)$ is an arbitrary trajectory in MFG emanating from the identity\cite{LoS}. For simplicity of notations, we omit superscript $u_t$ unless we emphasize the dependency of some variables on control inputs. By Theorem~\ref{th::mfg_automorphism}, all group affine dynamics are classified into \eqref{eq::general_gad_mfg}.
\begin{equation}
  \label{eq::general_gad_mfg}
  \Phi^t(\bs{\chi}_0)=\bs{S}\bs{\chi}_0\bs{S}^{-1}\Phi^t(id)
\end{equation}
Note $\bs{S}$ is defined in Theorem~\ref{th::mfg_automorphism}. A key observation is that $\bs{W}:=\bs{S}^{-1}\Phi^t(id)$ has a similar block-diagonal structure $\bs{W}=\diag{\bs{W}_1,...,\bs{W}_{s+t+1}}$ where each $\bs{W}_i$ takes value in $\text{SIM}_{n+m}(d)$. Since $\Phi^t(id)$ is an arbitrary chosen trajectory in MFG, the blocks $\bs{S}_i$ and $\bs{W}_i$ are related as \eqref{eq::relation_sw_blocks}
\begin{equation}
  \label{eq::relation_sw_blocks}
  \bs{S}_i=\begin{bmatrix}
    \bs{\Omega}_i & \bs{\varrho}_i\\
    \bs{0} & \bs{A}_i
  \end{bmatrix},\ \bs{W}_i=\begin{bmatrix}
    \bs{\Psi}_i & \bs{\varpi}_i\\
    \bs{0} & \bs{A}_i^{-1}
  \end{bmatrix}
\end{equation}
where $\bs{\Omega}_i,\bs{\Psi}_i\in\SO{d},\bs{\varrho}_i,\bs{\varpi}_i\in\mathbb{R}^{d\times(n+m)},\bs{A}_i\in\text{GL}(n+m.\mathbb{R})$ can be freely chosen as functions of input $u_t$. The lower-right block of $\bs{S}_i$ and $\bs{W}_i$ has a constraint to be inverse to each other. Taking derivative of \eqref{eq::general_gad_mfg} and letting $\bs{\chi}=\bs{S}\bs{\chi}_0\bs{W}$, we have the ODE \eqref{eq::mfg_gad_ode_total}.
\begin{gather}
  \label{eq::mfg_gad_ode_total}
  \begin{split}
    \dot{\bs{\chi}}&=\left(\dot{\bs{S}}\bs{S}^{-1}\right)\bs{S}\bs{\chi}_0\bs{W}+\bs{S}\bs{\chi}_0\bs{W}\left(\bs{W}^{-1}\dot{\bs{W}}\right)\\
    &=\left(\dot{\bs{S}}\bs{S}^{-1}\right)\bs{\chi}+\bs{\chi}\left(\bs{W}^{-1}\dot{\bs{W}}\right)
  \end{split}
\end{gather}
Let $\bs{\mathfrak{S}}=\dot{\bs{S}}\bs{S}^{-1}$ and $\bs{\mathfrak{W}}=\bs{W}^{-1}\dot{\bs{W}}$, both having diagonal structures with each block taking value in the Lie algebra $\mathfrak{sim}_{n+m}(d)$ of $\text{SIM}_{n+d}(d)$, denoted by \eqref{eq::sim_liealgebra_simple} similar to \eqref{eq::mfg_simple_notation}. We add superscript $l$ for $\SSS_0$ and $r$ for $\WWW_0$ for later convenience.
\begin{equation}
  \label{eq::sim_liealgebra_simple}
  \SSS=({^l}\SSS_0\bigl\vert{^l}\SSS_j\bigl\vert{^r}\SSS_j),\ \WWW=({^r}\WWW_0\bigl\vert{^l}\WWW_j\bigl\vert{^r}\WWW_j)
\end{equation}
Using the simplified notations in \eqref{eq::mfg_simple_notation} and \eqref{eq::sim_liealgebra_simple}, the process ODE writes \eqref{eq::mfg_gad_tfg1}--\eqref{eq::mfg_gad_tfg3} following \eqref{eq::mfg_gad_ode_total} in terms of TFG blocks.
\begin{align}
  \label{eq::mfg_gad_tfg1}
  \dot{\bs{T}}_0&={^l}\SSS_{0}\bs{T}_0+\bs{T}_0{^r}\WWW_{0}\\
  \label{eq::mfg_gad_tfg2}
  \frac{d}{dt}\left[{^l}\bs{T}_j\cdots \bs{T}_0\right]&={^l}\SSS_j{^l}\bs{T}_j\cdots\bs{T}_0+{^l}\bs{T}_j\cdots \bs{T}_0{^l}\WWW_j\\
  \label{eq::mfg_gad_tfg3}
  \frac{d}{dt}\left[\bs{T}_0\cdots {^r}\bs{T}_j\right]&={^r}\SSS_j\bs{T}_0\cdots {^r}\bs{T}_j+\bs{T}_0\cdots {^r}\bs{T}_j{^r}\WWW_j
\end{align}
Note the index $j$ for variables with superscript $l$ and $r$ runs from $1$ to $s$ and $t$ respectively. Simplifying \eqref{eq::mfg_gad_tfg2}--\eqref{eq::mfg_gad_tfg3} yields:
\begin{align}
  \label{eq::mfg_gad_tfg4}
  {^l}\dot{\bs{T}}_j&={^l}\SSS_j{^l}\bs{T}_j+{^l}\bs{T}_j\left(\text{Ad}_{{^l}\bar{\bs{T}}_{j-1}}{^l}\bar{\WWW}_j-{^l}\SSS_{j-1}\right)\\
  \label{eq::mfg_gad_tfg5}
  {^r}\dot{\bs{T}}_j&=\left(\text{Ad}^{-1}_{{^r}\bar{\bs{T}}_{j-1}}{^r}\bar{\SSS}_j-{^r}\WWW_{j-1}\right){^r}\bs{T}_j+{^r}\bs{T}_j{^r}\WWW_j
\end{align}
where we define notations ${^l}\bar{\WWW}_j={^l}\WWW_j-{^l}\WWW_{j-1}$, ${^r}\bar{\SSS}_j={^r}\SSS_j-{^r}\SSS_{j-1}$, ${^l}\bar{\bs{T}}_{j-1}={^l}\bs{T}_{j-1}{^l}\bs{T}_{j-2}\cdots\bs{T}_0$ and ${^r}\bar{\bs{T}}_{j-1}=\bs{T}_0\cdots{^r}\bs{T}_{j-2}{^r}\bs{T}_{j-1}$ in \eqref{eq::mfg_gad_tfg4}--\eqref{eq::mfg_gad_tfg5}\footnote{For index consistency when $j=0$, $\bs{T}_0={^l}\bs{T}_0={^r}\bs{T}_0$, ${^r}\bs{\SSS}_0={^l}\bs{\SSS}_0$ and ${^l}\bs{\WWW}_0={^r}\bs{\WWW}_0$. The same convention applies to their blocks.}. We must always keep in mind that the lower-right $(n+m)\times(n+m)$ block of ${^l}\bar{\WWW}_j$ and ${^l}\SSS_j-{^l}\SSS_{j-1}$ should be the same with a different sign by virtue of constraint \eqref{eq::relation_sw_blocks}, because the associated Lie algebras of two Lie group members which are inverse to each other differ with a minus sign. Similar arguments hold between ${^r}\bar{\SSS}_j$ and ${^r}\WWW_j-{^r}\WWW_{j-1}$. Let the components be \eqref{eq::gad_component_def1}--\eqref{eq::gad_component_def4}.
\begin{align}
  \label{eq::gad_component_def1}
  &{^{(\cdot)}}\bs{T}_j=\begin{bmatrix}
      {^{(\cdot)}}\bs{R}_j & {^{(\cdot)}}\bs{r}_j\\
      \bs{0} & \bs{I}
    \end{bmatrix},\ j\ge 0\\
  \label{eq::gad_component_def2}
  &{^l}\SSS_j=\begin{bmatrix}
      {^l}\bs{\theta}_j^\times & {^l}\bs{\gamma}_j\\
      \bs{0} & {^l}\bs{L}_j
    \end{bmatrix},\ {^r}\WWW_j=\begin{bmatrix}
      {^r}\bs{\omega}_j^\times & {^r}\bs{\rho}_j\\
      \bs{0} & {^r}\bs{L}_j
    \end{bmatrix},\ j\ge 0\\
  \label{eq::gad_component_def3}
  &{^l}\bar{\WWW}_j=\begin{bmatrix}
      {^l}\bs{\omega}_j^\times & {^l}\bs{\rho}_j\\
      \bs{0} & {^l}\bs{L}_{j-1}-{^l}\bs{L}_j
    \end{bmatrix},\ j\ge 1\\
  \label{eq::gad_component_def4}
  &{^r}\bar{\SSS}_j=\begin{bmatrix}
      {^r}\bs{\theta}_j^\times & {^r}\bs{\gamma}_j\\
      \bs{0} & {^r}\bs{L}_{j-1}-{^r}\bs{L}_j
    \end{bmatrix},\ j\ge 1
\end{align}
Note ${^{(\cdot)}}\bs{T}_j\in\TFG(d,n,m)$, the superscript $(\cdot)$ is $l, r$ or nothing. ${^l}\SSS_j,{^r}\WWW_j,{^l}\bar{\WWW}_j,{^r}\bar{\SSS}_j$ are in the Lie algebra $\mathfrak{sim}_{n+m}(d)$. ${^{(\cdot)}}\bs{R}_j\in\SO{d}$ and ${^{(\cdot)}}\bs{r}_j,{^l}\bs{\gamma}_j,{^r}\bs{\gamma}_j,{^l}\bs{\rho}_j,{^r}\bs{\rho}_j\in\mathbb{R}^{d\times(n+m)}$. Moreover, ${^l}\bs{\theta}_j,{^r}\bs{\theta}_j,{^l}\bs{\omega}_j,{^r}\bs{\omega}_j\in\mathbb{R}^{\frac{d(d-1)}{2}}$ and $(\cdot)^\times$ embeds an $\mathbb{R}^{\frac{d(d-1)}{2}}$ into the skew-symmetric matrix $\mathbb{R}^{d\times d}$. ${^{(\cdot)}}\bs{L}_j$ are arbitrary $\mathbb{R}^{(n+m)\times(n+m)}$ matrices as the Lie algebra $\mathfrak{gl}(n+m,\mathbb{R})$ of $\text{GL}(n+m,\mathbb{R})$, and specifically we have ${^l}\bs{L}_0=-{^r}\bs{L}_0$. Most importantly, \eqref{eq::gad_component_def2}--\eqref{eq::gad_component_def4} have already included the constraint \eqref{eq::relation_sw_blocks}, and thus every component with distinct notations (except ${^l}\bs{L}_0=-{^r}\bs{L}_0$) in the Lie algebra $\mathfrak{sim}_{n+m}(d)$ is a freely-chosen function of the control input $u_t$. Substituting \eqref{eq::gad_component_def1}--\eqref{eq::gad_component_def4} into \eqref{eq::mfg_gad_tfg1}, \eqref{eq::mfg_gad_tfg4}, \eqref{eq::mfg_gad_tfg5}, and equating blocks by brute calculation, we obtain (and thus have proved) the main theorem regarding the general form of group affine dynamics on a multi-frame group $\MFG(d,n,m,s,t)$.

\begin{theorem}\label{th::los_classification_dynamics}
  With the simplified notation for successive multiplication defined by \eqref{eq::successive_multiplication} and ${^{(\cdot)}}\bar{\bs{L}}_j:={^{(\cdot)}}\bs{L}_{j-1}-{^{(\cdot)}}\bs{L}_{j}$ where the superscript $(\cdot)$ is $l, r$ or nothing,
  \begin{equation}
    \label{eq::successive_multiplication}
    {^{(\cdot)}}\bar{\bs{R}}_{i,j}=\begin{cases}
    {^{(\cdot)}}\bs{R}_{i}{^{(\cdot)}}\bs{R}_{i+1}\cdots{^{(\cdot)}}\bs{R}_{j-1}{^{(\cdot)}}\bs{R}_{j}, & i < j\\
    {^{(\cdot)}}\bs{R}_{i}, & i = j\\
    {^{(\cdot)}}\bs{R}_{i}{^{(\cdot)}}\bs{R}_{i-1}\cdots{^{(\cdot)}}\bs{R}_{j+1}{^{(\cdot)}}\bs{R}_{j}, & i > j
    \end{cases}
  \end{equation}
  then all group affine dynamics on a multi-frame group are classified into the frame-related ODEs \eqref{eq::gad_mfg_frame1}--\eqref{eq::gad_mfg_frame3}:
  \begin{align}
    \label{eq::gad_mfg_frame1}
    \dot{\bs{R}}_0&={^l}\bs{\theta}_0^\times\bs{R}_0+\bs{R}_0{^r}\bs{\omega}_0^\times\\
    \label{eq::gad_mfg_frame2}
    {^l}\dot{\bs{R}}_j&={^l}\bs{\theta}_j^\times{^l}\bs{R}_j+{^l}\bs{R}_j\left[\text{Ad}_{{^l}\bar{\bs{R}}_{j-1,0}}({^l}\bs{\omega}_j^\times)-{^l}\bs{\theta}_{j-1}^\times\right]\\
    \label{eq::gad_mfg_frame3}
    {^r}\dot{\bs{R}}_j&=\left[\text{Ad}^{-1}_{{^r}\bar{\bs{R}}_{0,j-1}}({^r}\bs{\theta}_j^\times)-{^r}\bs{\omega}_{j-1}^\times\right]{^r}\bs{R}_j+{^r}\bs{R}_j{^r}\bs{\omega}_j^\times
  \end{align}
  and vector-related ODEs \eqref{eq::gad_mfg_vec1}--\eqref{eq::gad_mfg_vec3}:
  \begin{align}
    \label{eq::gad_mfg_vec1}
    &\dot{\bs{r}}_0={^l}\bs{\theta}_0^\times\bs{r}_0+{^l}\bs{\gamma}_0+\bs{R}_0{^r}\bs{\rho}_0+\bs{r}_0{^r}\bs{L}_0\\
    \notag
    &{^l}\dot{\bs{r}}_j=\sum\limits_{k=0}^{j-1}{^l}\bar{\bs{R}}_{j,k+1}\left[{^l}\bs{r}_k{^l}\bar{\bs{L}}_{j}-\text{Ad}_{{^l}\bar{\bs{R}}_{k,0}}({^l}\bs{\omega}_j^\times){^l}\bs{r}_k\right]+{^l}\bs{\gamma}_j\\
    \label{eq::gad_mfg_vec2}
    &\ +{^l}\bar{\bs{R}}_{j,0}{^l}\bs{\rho}_j+{^l}\bs{\theta}_j^\times{^l}\bs{r}_j-{^l}\bs{R}_j{^l}\bs{\gamma}_{j-1}-{^l}\bs{r}_j{^l}\bs{L}_{j-1}\\
    \notag
    &{^r}\dot{\bs{r}}_j={^r}\bar{\bs{R}}_{0,j-1}^{-1}{^r}\bs{\theta}_j^\times\bs{r}_0+\sum\limits_{k=1}^{j-1}{^r}\bar{\bs{R}}_{k,j-1}^{-1}\text{Ad}^{-1}_{{^r}\bar{\bs{R}}_{0,k-1}}({^r}\bs{\theta}_j^\times){^r}\bs{r}_k\\
    \notag
    &\ +\text{Ad}^{-1}_{{^r}\bar{\bs{R}}_{0,j-1}}({^r}\bs{\theta}_j^\times){^r}\bs{r}_j-{^r}\bs{\rho}_{j-1}-\sum\limits_{k=0}^{j-1}{^r}\bar{\bs{R}}_{k,j-1}^{-1}{^r}\bs{r}_k{^r}\bar{\bs{L}}_{j}\\
    \label{eq::gad_mfg_vec3}
    &\ +{^r}\bs{R}_j{^r}\bs{\rho}_j+{^r}\bs{r}_j{^r}\bs{L}_j-{^r}\bs{\omega}_{j-1}^\times{^r}\bs{r}_j+{^r}\bar{\bs{R}}_{0,j-1}^{-1}{^r}\bs{\gamma}_j
  \end{align}
  where ${^{(\cdot)}}\bs{R}_{(\cdot)},{^{(\cdot)}}\bs{r}_{(\cdot)}$ are MFG state components and all other variables are arbitrary chosen functions of the control input $u_t$ as described above. The sizes of the matrix of those variables have been defined in preceding paragraphs.
\end{theorem}

\begin{remark}\label{rm::gad_dynamics}
  Every ${^{(\cdot)}}\bs{r}_j$ is composed of two types of vectors associated to one of the two frames that a TFG transformation connects. We could further break the following variables into column blocks: ${^{(\cdot)}}\bs{r}_j=\left[{^{(\cdot)}}\bs{x}_j,\ {^{(\cdot)}}\bs{R}_j{^{(\cdot)}}\bs{y}_j\right]$, ${^{(\cdot)}}\bs{\gamma}=\left[{^{(\cdot)}}\bs{\alpha}_j,\ {^{(\cdot)}}\bs{\beta}_j\right]$, ${^{(\cdot)}}\bs{\rho}_j=\left[{^{(\cdot)}}\bs{\xi}_j,\ {^{(\cdot)}}\bs{\eta}_j\right]$ and ${^{(\cdot)}}\bs{L}_j=\begin{bmatrix}{^{(\cdot)}}\bs{L}_{A,j} & {^{(\cdot)}}\bs{L}_{B,j}\\ {^{(\cdot)}}\bs{L}_{C,j} & {^{(\cdot)}}\bs{L}_{D,j}\end{bmatrix}$, where we have ${^{(\cdot)}}\bs{x}_j,{^{(\cdot)}}\bs{\alpha}_j,{^{(\cdot)}}\bs{\xi}_j\in\mathbb{R}^{d\times n}$, ${^{(\cdot)}}\bs{y}_j,{^{(\cdot)}}\bs{\beta}_j,{^{(\cdot)}}\bs{\eta}_j\in\mathbb{R}^{d\times m}$. Besides, ${^{(\cdot)}}\bs{L}_{A,j}\in\mathbb{R}^{n\times n}$, ${^{(\cdot)}}\bs{L}_{B,j}\in\mathbb{R}^{n\times m}$, ${^{(\cdot)}}\bs{L}_{C,j}\in\mathbb{R}^{m\times n}$, ${^{(\cdot)}}\bs{L}_{D,j}\in\mathbb{R}^{m\times m}$, and $(\cdot)$ is $l, r$ or nothing. Substituting the above into \eqref{eq::gad_mfg_vec1}--\eqref{eq::gad_mfg_vec3} we will obtain equations in terms of ${^{(\cdot)}}\bs{x}_j$ and ${^{(\cdot)}}\bs{y}_j$. Due to space limitations, we only show results related to $j=0$.
  \begin{align}
    \label{eq::recover_tfg_gad1}
    \dot{\bs{x}}_0&={^l}\bs{\theta}_0^\times\bs{x}_0+{^l}\bs{\alpha}_0+\bs{R}_0\left[{^r}\bs{\xi}_0+\bs{y}_0{^r}\bs{L}_{C,0}\right]+\bs{x}_0{^r}\bs{L}_{A,0}\\
    \label{eq::recover_tfg_gad2}
    \dot{\bs{y}}_0&=-{^r}\bs{\omega}_{0}^\times\bs{y}_0+\bs{R}_0^{-1}\left[{^l}\bs{\beta}_{0}+\bs{x}_0{^r}\bs{L}_{B,0}\right]+{^r}\bs{\eta}_{0}+\bs{y}_0{^r}\bs{L}_{D,0}
  \end{align}
  Note that \eqref{eq::recover_tfg_gad1}--\eqref{eq::recover_tfg_gad2} exactly recover the `natural vector dynamics' in \cite{TFG} as a complete classification of all possible vector dynamics being group affine under TFG structure. 
\end{remark}

\subsection{Classification of Algebraic Observations on MFGs}

The algebraic observation can be classified by exploring how MFG could act on a vector space. As MFG has been embedded into a block-diagonal matrix, the natural left action on a constant vector is the matrix multiplication. Consider $\tilde{\bs{\chi}}\in\TFG$ to be one of the diagonal blocks of MFG, i.e. one of $\bs{T}_0$, $[{^l}\bs{T}_j\cdots\bs{T}_0]$ or $[\bs{T}_0\cdots{^r}\bs{T}_j]$, and let two known vectors be $\bar{\bs{d}},\underline{\bs{d}}\in V:=\mathbb{R}^{m+n+d}$. Denote $\bar{\bs{w}},\underline{\bs{w}}\in V$ to be the vector-valued measurements respectively. A linear observed measurement equation writes \eqref{eq::obs_total}.
\begin{equation}
  \label{eq::obs_total}
  \bar{\bs{w}}=\tilde{\bs{\chi}}\bar{\bs{d}},\quad\underline{\bs{w}}=\tilde{\bs{\chi}}^{-1}\underline{\bs{d}}  
\end{equation}
Matrix multiplication of $\tilde{\bs{\chi}}$ or $\tilde{\bs{\chi}}^{-1}$ corresponds to the left or right action, compatible with the left- or right-invariant error respectively\cite{InEKF}. Other actions on $\bar{\bs{d}},\underline{\bs{d}}$ can be generated by applying the automorphism on $\tilde{\bs{\chi}}$, as Lemma~\ref{lm::group_action}.

\begin{lemma}\label{lm::group_action}
  Let $M$ be a manifold and $G$ be a group acting on $M$. Let the left action be $\blacktriangleright:G\times M\rightarrow M,\ (g,p)\mapsto g\blacktriangleright p$. Let $\phi\in\Aut(G)$, and we get a new left action from the old one as $\rhd:G\times M\rightarrow M,\ (g,p)\mapsto \phi(g)\blacktriangleright p$.
\end{lemma}
\begin{proof}
  Let $p\in M$ and $g_1,g_2\in G$. First, $id\rhd p=\phi(id)\blacktriangleright p=id\blacktriangleright p=p$. Moreover, $g_2\rhd(g_1\rhd p)=\phi(g_2)\blacktriangleright(\phi(g_1)\blacktriangleright p)=(\phi(g_2)\phi(g_1))\blacktriangleright p=\phi(g_2g_1)\blacktriangleright p=(g_2g_1)\rhd p$, proving by definition of the group action.
\end{proof}

The right action case is treated accordingly. Let the state $\tilde{\bs{\chi}}$ be $\begin{bmatrix}\tilde{\bs{R}} & \tilde{\bs{r}}\\ \bs{0} & \bs{I}\end{bmatrix}$. Combining \eqref{eq::obs_total} and Lemma~\ref{lm::group_action} yield \eqref{eq::obs_matrix_form1} and \eqref{eq::obs_matrix_form2}, partitioning $\bar{\bs{w}},\underline{\bs{w}}$ and $\bar{\bs{d}},\underline{\bs{d}}$ into row blocks $\bar{\bs{w}}^{(1)},\underline{\bs{w}}^{(1)},\bar{\bs{d}}^{(1)},\underline{\bs{d}}^{(1)}\in\mathbb{R}^d$, $\bar{\bs{w}}^{(2)},\underline{\bs{w}}^{(2)},\bar{\bs{d}}^{(2)},\underline{\bs{d}}^{(2)}\in\mathbb{R}^{n+m}$.
\begin{align}
  \label{eq::obs_matrix_form1}
  \begin{bmatrix}\bar{\bs{w}}^{(1)} \\ \bar{\bs{w}}^{(2)}\end{bmatrix}&=\begin{bmatrix}\bar{\bs{\Omega}} & \bar{\bs{\varrho}} \\ \bs{0} & \bar{\bs{A}}\end{bmatrix}\begin{bmatrix}
    \tilde{\bs{R}} & \tilde{\bs{r}}\\
    \bs{0} & \bs{I}
  \end{bmatrix}\begin{bmatrix}\bar{\bs{\Omega}} & \bar{\bs{\varrho}}\\ \bs{0} & \bar{\bs{A}}\end{bmatrix}^{-1}\begin{bmatrix}\bar{\bs{d}}^{(1)} \\ \bar{\bs{d}}^{(2)}\end{bmatrix}\\
  \label{eq::obs_matrix_form2}
  \begin{bmatrix}\underline{\bs{w}}^{(1)} \\ \underline{\bs{w}}^{(2)}\end{bmatrix}&=\begin{bmatrix}\underline{\bs{\Omega}} & \underline{\bs{\varrho}}\\ \bs{0} & \underline{\bs{A}}\end{bmatrix}\begin{bmatrix}
      \tilde{\bs{R}} & \tilde{\bs{r}}\\
      \bs{0} & \bs{I}
  \end{bmatrix}^{-1}\begin{bmatrix}\underline{\bs{\Omega}} & \underline{\bs{\varrho}}\\ \bs{0} & \underline{\bs{A}}\end{bmatrix}^{-1}\begin{bmatrix}\underline{\bs{d}}^{(1)} \\ \underline{\bs{d}}^{(2)}\end{bmatrix}
\end{align}
The $\bar{\bs{\Omega}},\underline{\bs{\Omega}}\in\SO{d}$, $\bar{\bs{\varrho}},\underline{\bs{\varrho}}\in\mathbb{R}^{d\times(n+m)}$ and $\bar{\bs{A}},\underline{\bs{A}}\in\text{GL}(n+m,\mathbb{R})$ are fixed constants in their respective spaces. We are only interested in the first row block of \eqref{eq::obs_matrix_form1} or \eqref{eq::obs_matrix_form2}, similar to the case where we only take the first three rows when considering $\text{SE}(3)$ acting on $\mathbb{R}^3$ by matrix multiplication of $\mathbb{R}^{4\times 4}$ on a homogeneous vector $\mathbb{R}^4$. Letting $\tilde{\bs{\chi}}$ be one of $\bs{T}_0$, $[{^l}\bs{T}_j\cdots\bs{T}_0]$ or $[\bs{T}_0\cdots{^r}\bs{T}_j]$ and simplifying the equations obtained by equating the first row blocks of \eqref{eq::obs_matrix_form1} and \eqref{eq::obs_matrix_form2}, we have proved the following theorem.

\begin{theorem}\label{th::los_classification_obs}
  Inherit simplified notations from \eqref{eq::successive_multiplication}, and let $\bar{\bs{w}}_0,{^{(\cdot)}}\bar{\bs{w}}_j,\underline{\bs{w}}_0,{^{(\cdot)}}\underline{\bs{w}}_j$ be vector observations in $\mathbb{R}^d$. Let $\bar{\bs{d}}^{(1)}_0,{^{(\cdot)}}\bar{\bs{d}}^{(1)}_j$ be constant $\mathbb{R}^d$. Let $\bar{\bs{d}}^{(2)}_0,{^{(\cdot)}}\bar{\bs{d}}^{(2)}_j$ be constant $\mathbb{R}^{n+m}$. Let $\bar{\bs{\Omega}}_0,{^{(\cdot)}}\bar{\bs{\Omega}}_j,\underline{\bs{\Omega}}_0,{^{(\cdot)}}\underline{\bs{\Omega}}_j$ be constant $\SO{d}$. Let $\bar{\bs{\varrho}}_0,{^{(\cdot)}}\bar{\bs{\varrho}}_j,\underline{\bs{\varrho}}_0,{^{(\cdot)}}\underline{\bs{\varrho}}_j$ be constant $\mathbb{R}^{d\times(n+m)}$. Let the superscript $(\cdot)$ be $l$ or $r$. Those constants are arbitrarily set in their respective domains. Notations for MFG components ${^{(\cdot)}}\bs{R}_{(\cdot)},{^{(\cdot)}}\bs{r}_{(\cdot)}$ follow Theorem~\ref{th::los_classification_dynamics}. Moreover, introduce newly-defined notations ${^l}\bar{\bs{r}}_j={^l}\bs{r}_j+\sum_{k=0}^{j-1}{^l}\bar{\bs{R}}_{j,k+1}{^l}\bs{r}_k$, ${^r}\bar{\bs{r}}_j=\bs{r}_0+\sum_{k=1}^j{^r}\bar{\bs{R}}_{0,k-1}{^r}\bs{r}_k$, ${^l}\underline{\bs{r}}_j=\sum_{k=0}^{j}{^l}\bar{\bs{R}}_{k,0}^{-1}{^l}\bs{r}_k$ and ${^r}\underline{\bs{r}}_j=\sum_{k=0}^j {^r}\bar{\bs{R}}_{k,j}^{-1}{^r}\bs{r}_k$. The algebraic observation of linear observed systems on multi-frame groups can be classified into the forms of left-action case \eqref{eq::obs_mfg_component1}--\eqref{eq::obs_mfg_component3}
  \begin{align}
    \label{eq::obs_mfg_component1}
    &\bar{\bs{w}}_0=\left(\mathcal{C}_{\bar{\bs{\Omega}}_0}\bs{R}_0\right)\bar{\bs{d}}_{0}^{(1)}+\left(\bar{\bs{\Omega}}_0\bs{r}_0+\bar{\bs{\varrho}}_0\right)\bar{\bs{d}}_0^{(2)}\\
    \label{eq::obs_mfg_component2}
    &{^l}\bar{\bs{w}}_j=\left(\mathcal{C}_{{^l}\bar{\bs{\Omega}}_j}{^l}\bar{\bs{R}}_{j,0}\right){^l}\bar{\bs{d}}_{j}^{(1)}+\left({^l}\bar{\bs{\varrho}}_j+{^l}\bar{\bs{\Omega}}_j{^l}\bar{\bs{r}}_j\right){^l}\bar{\bs{d}}_j^{(2)}\\
    \label{eq::obs_mfg_component3}
    &{^r}\bar{\bs{w}}_j=\left(\mathcal{C}_{{^r}\bar{\bs{\Omega}}_j}{^r}\bar{\bs{R}}_{0,j}\right){^r}\bar{\bs{d}}_{j}^{(1)}+\left({^r}\bar{\bs{\varrho}}_j+{^r}\bar{\bs{\Omega}}_j{^r}\bar{\bs{r}}_j\right){^r}\bar{\bs{d}}_j^{(2)}
  \end{align}
  and right-action case \eqref{eq::obs_mfg_component4}--\eqref{eq::obs_mfg_component6}.
  \begin{align}
    \label{eq::obs_mfg_component4}
    &\underline{\bs{w}}_0=\left(\mathcal{C}_{\underline{\bs{\Omega}}_0}\bs{R}^{-1}_0\right)\underline{\bs{d}}_0^{(1)}+\left(\underline{\bs{\varrho}}_0-\underline{\bs{\Omega}}_0\bs{R}^{-1}_0\bs{r}_0\right)\underline{\bs{d}}_0^{(2)}\\
    \label{eq::obs_mfg_component5}
    &{^l}\underline{\bs{w}}_j=\left(\mathcal{C}_{{^l}\underline{\bs{\Omega}}_j}{^l}\bar{\bs{R}}_{j,0}^{-1}\right){^l}\underline{\bs{d}}_j^{(1)}+\left({^l}\underline{\bs{\varrho}}_j-{^l}\underline{\bs{\Omega}}_j{^l}\underline{\bs{r}}_j\right){^l}\underline{\bs{d}}_j^{(2)}\\
    \label{eq::obs_mfg_component6}
    &{^r}\underline{\bs{w}}_j=\left(\mathcal{C}_{{^r}\underline{\bs{\Omega}}_j}{^r}\bar{\bs{R}}_{0,j}^{-1}\right){^r}\underline{\bs{d}}_j^{(1)}+\left({^r}\underline{\bs{\varrho}}_j-{^r}\underline{\bs{\Omega}}_j {^r}\underline{\bs{r}}_j\right){^r}\underline{\bs{d}}_j^{(2)}
  \end{align}
\end{theorem}
\begin{remark}
  Similarly, we could break the state into column blocks, e.g. ${^{(\cdot)}}\bs{r}_j=\left[{^{(\cdot)}}\bs{x}_j,\ {^{(\cdot)}}\bs{R}_j{^{(\cdot)}}\bs{y}_j\right]$ to reproduce `natural two-frame outputs'\cite{TFG} from \eqref{eq::obs_mfg_component1} and \eqref{eq::obs_mfg_component4}.
\end{remark}

\section{Application in Navigation}\label{sec::application}

Depth-camera inertial odometry with simultaneous estimation of camera extrinsics involving three frames is proposed to show the application of MFG. Let $\bs{R}\in\SO{3}$ be the rotation from the IMU frame to the world frame. $\bs{p},\bs{v}\in\mathbb{R}^3$ denote IMU position and velocity. Let $(\bs{R}_c,\bs{p}_c)\in\text{SE}(3)$ be the extrinsic parameter from the depth camera frame to the IMU frame. The IMU and camera are both rigidly attached to a moving vehicle. We estimate all the above states simultaneously. The process ODE with noise are \eqref{eq::sim_process_ode1}--\eqref{eq::sim_process_ode2},
\begin{align}
  \label{eq::sim_process_ode1}
  &\dot{\bs{R}}=\bs{R}(\bs{\omega}+\bs{n}_g)^\times,\ \dot{\bs{R}}_c=\bs{R}_c\bs{n}_c^\times\\
  \label{eq::sim_process_ode2}
  &\dot{\bs{p}}=\bs{v},\ \dot{\bs{v}}=\bs{R}(\bs{a}+\bs{n}_a)+\bs{g},\ \dot{\bs{p}}_c=\bs{n}_p
\end{align}
where $\bs{\omega},\bs{a}\in\mathbb{R}^3$ are gyroscope and accelerometer inputs. $\bs{n}_g,\bs{n}_a,\bs{n}_c,\bs{n}_p$ are $\mathbb{R}^3$-valued zero-mean Gaussian noise. Note we also associate noise to the extrinsic parameters. Our depth-camera measures the 3D position $\bs{z}_j\in\mathbb{R}^3(j=1,2,3)$ of three landmarks with respect to the camera frame with known global position $\bs{p}_{l,j}$. The observation equation with zero-mean Gaussian noise $\bs{n}_{z,j}\in\mathbb{R}^3$ writes \eqref{eq::sim_algebraic_obs}.
\begin{equation}
  \label{eq::sim_algebraic_obs}
  \bs{z}_j=\bs{R}_c^{-1}\left[\bs{R}^{-1}(\bs{p}_{l,j}-\bs{p})-\bs{p}_c\right]+\bs{n}_{z,j},\ j=1,2,3
\end{equation}

Non-biased IMU dynamics is considered here, though biases can be involved in TFG structure\cite{TFG}. More advanced geometrical treatments on biases are thoroughly explored in \cite{af_thesis}. \cite{att_online_ext} tackles online extrinsic calibration via semi-direct product of attitude $\SO{d}$ only. Traditionally, we may endow $\SO{3}\times\mathbb{R}^6\times\SO{3}\times\mathbb{R}^3$ or $\text{SE}_2(3)\times\text{SE}(3)$ on the state $(\bs{R},\bs{p},\bs{v},\bs{R}_c,\bs{p}_c)$, inducing a multiplicative EKF (MEKF) or an imperfect IEKF respectively. There is no coupling between the core IMU pose and the camera extrinsics due to the direct-product structure in the state. One may notice that \eqref{eq::sim_algebraic_obs} without noise can be cast into the form of \eqref{eq::obs_mfg_component6} in Theorem~\ref{th::los_classification_obs} with $j=1$, and thus $\MFG(3,2,0,0,1)$ can be endowed upon the state making \eqref{eq::sim_algebraic_obs} linear observed. Meanwhile, the process ODEs \eqref{eq::sim_process_ode1}--\eqref{eq::sim_process_ode2} are no longer group affine by Theorem~\ref{th::los_classification_dynamics}. This is not an obstruction because the $\MFG(3,2,0,0,1)$ construction reduces linearization error in \eqref{eq::sim_algebraic_obs} and transfers the nonlinearities in \eqref{eq::sim_algebraic_obs} to slow dynamics $\dot{\bs{R}}_c=\bs{R}_c\bs{n}_c^\times$ and $\dot{\bs{p}}_c=\bs{n}_p$, similar to the philosophies in tackling biases\cite{af_thesis}. Let the true state $\bs{\chi}=(\bs{R},\bs{p},\bs{v},\bs{R}_c,\bs{p}_c)$ be $\MFG(3,2,0,0,1)$. Let the error-state be $\delta\bs{\chi}=[\delta\bs{\theta};\delta\bs{p};\delta\bs{v};\delta\bs{\theta}_c;\delta\bs{p}_c]\in\mathbb{R}^{15}$ where each vector block is $\mathbb{R}^3$. $\hat{\bs{\chi}}\in\MFG(3,2,0,0,1)$ denotes the estimated state. The MFG nonlinear update is then realized from \eqref{eq::sim_nonlinear_update1}--\eqref{eq::sim_nonlinear_update3}.
\begin{align}
  \label{eq::sim_nonlinear_update1}
  &\bs{R}=\hat{\bs{R}}\Gamma_0(\delta\bs{\theta}),\ \bs{R}_c=\Gamma_0(\hat{\bs{R}}^{-1}\delta\bs{\theta}_c)\hat{\bs{R}}_c\\
  \label{eq::sim_nonlinear_update2}
  &\bs{p}=\Gamma_0(\delta\bs{\theta})\hat{\bs{p}}+\Gamma_1(\delta\bs{\theta})\delta\bs{p},\bs{v}=\Gamma_0(\delta\bs{\theta})\hat{\bs{v}}+\Gamma_1(\delta\bs{\theta})\delta\bs{v}\\
  \label{eq::sim_nonlinear_update3}
  &\bs{p}_c=\Gamma_0(\hat{\bs{R}}^{-1}\delta\bs{\theta}_c)\hat{\bs{p}}_c+\hat{\bs{R}}^{-1}\Delta(\delta\bs{\theta}_c,\hat{\bs{p}},\delta\bs{p}_c)
\end{align}
where the auxiliary functions are $\Gamma_m(\delta\bs{\theta})=\sum_{n=0}^{+\infty}\frac{(\delta\bs{\theta})^\times}{(n+m)!}$ and $\Delta(\delta\bs{\theta}_c,\hat{\bs{p}},\delta\bs{p}_c)=(\Gamma_0(\delta\bs{\theta}_c)-\bs{I})\hat{\bs{p}}+\Gamma_1(\delta\bs{\theta}_c)\delta\bs{p}_c$. Update formulas for MEKF and IEKF (imperfect) are omitted.

\begin{figure}[t]%\color{blue}
  \centering
  \parbox{3.1in}{
    \centering
    \includegraphics[scale=0.70]{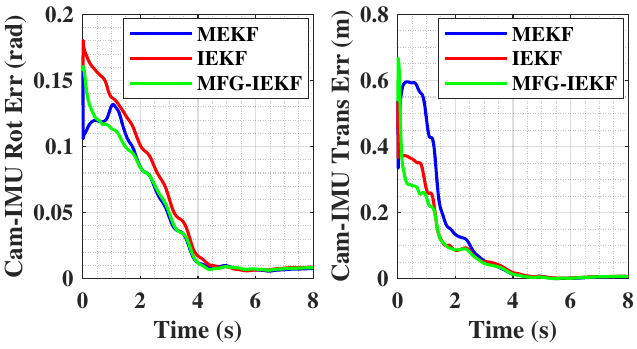}
   }
   \caption{Plot of error w.r.t time of MEKF, imperfect IEKF (IEKF) and our proposed MFG-IEKF. Only extrinsics (Camera to IMU) states are demonstrated. `Rot Err' and `Trans Err' are short for rotational error $\Vert\log(\hat{\bs{R}}_c\bs{R}_c^{-1})\Vert$ and translational error $\Vert\hat{\bs{p}}_c-\bs{p}_c\Vert$.}
   \label{fig::error_plot}
\end{figure}

Simulation to compare MEKF, imperfect IEKF and our proposed MFG-IEKF is shown in Fig.~\ref{fig::error_plot} treating those filters as deterministic observers, i.e. noise-free propagation and observation with all covariances of $\bs{n}_g,\bs{n}_c,\bs{n}_a,\bs{n}_p,\bs{n}_{z,j}$ as tunable gains\cite{InEKF}. The state trajectory are preset analytical functions and landmarks are at constant known positions. MEKF, IEKF and MFG-IEKF share the same tunable parameters (covariances) and initial states with deliberate offsets on camera-IMU extrinsics $(\bs{R}_c,\bs{p}_c)$. It's not surprising to see our MFG-IEKF has better transient performance due to more exact linearization as explained above. Meanwhile, all filters have similar steady-state behaviors. Afterall, all the above filters can be realized under a unified equivariant filter framework with different linearization approaches induced from different symmetries\cite{af_thesis}. 

\section{Conclusion}\label{sec::conclusion}

We construct multi-frame group for simultaneous navigational estimation involving multiple frames by semi-direct product covering natural extensions. All linear observed systems including group affine ODEs and algebraic observations are classified via the automorphism group of MFG. An example is provided to show the advantage of MFG-IEKF.
 
% \section*{APPENDIX}

% \subsection{Proof of Theorem xxx}

% \begin{proof}
%   Some excellent proof. 
% \end{proof}

% {\noindent\hspace{2em}{\itshape Proof of Theorem xxx:}} Another excellent proof.
% \endproof

% \addtolength{\textheight}{-12cm}   % This command serves to balance the column lengths
                                  % on the last page of the document manually. It shortens
                                  % the text height of the last page by a suitable amount.
                                  % This command does not take effect until the next page
                                  % so it should come on the page before the last. Make
                                  % sure that you do not shorten the text height too much.

%%%%%%%%%%%%%%%%%%%%%%%%%%%%%%%%%%%%%%%%%%%%%%%%%%%%%%%%%%%%%%%%%%%%%%%%%%%%%%%%

%%%%%%%%%%%%%%%%%%%%%%%%%%%%%%%%%%%%%%%%%%%%%%%%%%%%%%%%%%%%%%%%%%%%%%%%%%%%%%%%

%%%%%%%%%%%%%%%%%%%%%%%%%%%%%%%%%%%%%%%%%%%%%%%%%%%%%%%%%%%%%%%%%%%%%%%%%%%%%%%%

% \section*{ACKNOWLEDGMENT}

% Thanks everyone.

% \begin{thebibliography}
% \end{thebibliography}
\bibliographystyle{IEEEtran.bst}
% \bibColoredItems{black}{}
\bibliography{refs.bib}

\end{document}